\documentclass{cms}

\usepackage{amsmath}
\usepackage{amsfonts}
\usepackage{amssymb}
\usepackage{upref}
\usepackage[noadjust]{cite}
\usepackage[oldenum,flushleft]{paralist}
\usepackage{cr}
\usepackage{eucal}
\usepackage{amscd}
\usepackage[all,cmtip]{xy}
\usepackage{rotating}
\usepackage{hyperref,mathrsfs}

\newcommand{\vo}{\vec{o}\@ifnextchar{^}{\,}{}}
\makeatother

\def\eps{\epsilon}%
\def\tensor{\,\raise2pt\hbox{${}_{\otimes}$}\,}
\def\fdg{\,:\,}
\def\ptl{\partial}
\def\rest#1{\raise-2pt\hbox{${\lfloor_{#1}}$}}

\def\ulin#1{\underline{#1}{}}

\def\grad{{\nabla}}
\newcommand{\leftexp}[2]{{\vphantom{#2}}^{#1}{#2}}

\def\halb{\frac{1}{2}}

\def \a{\alpha}
\def \b {\beta}

\newtheorem{theorem}{Theorem}[section]

\newtheorem{proposition}[theorem]{Proposition}

\newcommand{\ba}{\begin{array}}
\newcommand{\ea}{\end{array}}
\newcommand{\bea}{\begin{eqnarray}}
\newcommand{\eea}{\end{eqnarray}}
\newcommand{\bee}{\begin{eqnarray*}}
\newcommand{\eee}{\end{eqnarray*}}

\renewcommand{\c}{\cdot}

\renewcommand{\gg}{{\bf g}}

\renewcommand{\a}{\alpha}
\renewcommand{\b}{\beta}

\renewcommand{\r}{\rho}

\begin{document}

\vol{40}
\num{2}
\volyear{2018}
\refnumber{P485}

\startpage{39}
\numberofpages{16}
\received{2017-11-04}
\revised{2018-04-06}

\title[Maxwell's Equations on Kerr-de Sitter Black Holes]{ A Positive-Definite Energy Functional for Axially Symmetric Maxwell's Equations on  Kerr-de Sitter Black Hole Spacetimes}


\presentedby{Niky Kamran}
\author{Nishanth Gudapati}
\address{Department of Mathematics, Yale University, 10 Hillhouse Avenue, New Haven, CT-06511, USA}

\email{nishanth.gudapati@yale.edu}

\amsclassification{83C50}{35L65}
\keywords{Kerr-de Sitter Black Holes, Stability of Black Holes, Wave Maps}


\maketitle

\begin{abstract}
We prove that there exists a phase space of canonical variables, for the initial value problem for axially symmetric Maxwell fields with compactly supported initial data and propagating in  Kerr-de Sitter black hole spacetimes, such that their motion is restricted to the level sets of a positive-definite Hamiltonian, despite the ergo-region.  
\end{abstract}

\french

\begin{resume}
On d\'emontre qu'il existe un espace de phase de variables
canoniques, pour le probl\`eme des valeurs initials pour les champs de
Maxwell sym\'etriques \`a donne\'es initiales de support compact et
\`a propagation dans les espaces-temps de trou noir Kerr-de Sitter,
tel que leur motion est restrainte aux ensembles de niveau
d'une hamiltonienne de type positif, en d\'epit de l'ergo-r\'egion.
\end{resume}

\english

\section{Kerr-de Sitter Black Holes}

Consider the Kerr-de Sitter family of black holes $(\bar{M}, \bar{g}):$
\begin{align}\label{KdS1}
\bar{g}=& -\frac{\Delta}{\Sigma} \left(\frac{dt - a\sin^2 \theta d\phi}{1+ \frac{\Lambda}{3}a^2} \right)^2 + \frac{\Sigma}{\Delta} dr^2 + \frac{\Sigma}{\Pi} d\theta^2  \notag\\
&+ 
\frac{\sin^2 \theta (1+ \frac{\Lambda}{3} a^2 \cos^2\theta)}{\Sigma}  \left(\frac{adt- (r^2 +a^2) d\phi}{1+ \frac{\Lambda}{3} a^2} \right)^2
\end{align}
where 
\begin{subequations}
\begin{align}
\Delta =& r^2 -2mr + a^2 -\frac{\Lambda r^2}{3} (r^2 + a^2), \\
\Sigma=& r^2 + a^2 \cos^2 \theta, \\
\Pi =& 1+ \frac{\Lambda}{3} a^2 \cos^2 \theta,
\end{align}
\end{subequations}
$\theta \in [0, \pi],$ $\phi \in [0, 2\pi)$, $\vert a \vert <m.$ The quartic polynomial function $\Delta(r)$ is such that it admits precisely one negative root and three distinct positive roots $\{ r_{\pm}, r_c\}, r_+ < r_c$, which correspond to the cases of physical interest. In this work we shall restrict to $a \neq 0$ and the regular region $r_+ < r < r_c.$ The Kerr-de Sitter family is a solution of Einstein's equations in 3+1 dimensions with a positive cosmological constant: 
\begin{align}\label{ee}
\bar{R}_{\mu \nu} - \halb \bar{g}_{\mu \nu} R_{\bar{g}} + \Lambda \bar{g}_{\mu \nu} =0, \quad \Lambda > 0\quad (\bar{M}, \bar{g}),
\end{align}
\noindent which reduces to the Schwarzschild-de Sitter family if $a =0$ and de Sitter if $a$ and $ m =0.$   The question of stability of the 3+1 de Sitter spacetime has been resolved by Friedrich in a series of landmark works \cite{F86, F86_2,F91}. Existence and stability of even dimensional de Sitter spacetimes in higher dimensions was proved in \cite{M05}. In a remarkable recent breakthrough, Hintz and Vasy have resolved the \emph{nonlinear} stability of the Kerr-de Sitter black holes for small angular-momentum \cite{HV_16} (see also \cite{H_16}). 

The evolution of their methods, developed from Melrose's $b-$calculus, can be found in the list of references therein.  In this context, there were preceding results on the stability of the Kerr-de Sitter family for small angular-momentum for various model problems. The local energy decay of the wave equation on Schwarzschild-de Sitter is studied in \cite{JFH_08}. Asymptotics and resonances of linear waves on the Kerr-de Sitter metric are studied in \cite{SD12, SD15}. The asymptotic behaviour of the Klein-Gordon equation on the Kerr-de Sitter metric is studied in \cite{GGH_17}. Global boundedness for linear waves on Schwarzschild-de Sitter and Kerr-de Sitter cosmologies is proved in \cite{VS12}. A partial proof of \emph{nonlinear} stability of Schwarzschild-de Sitter cosmologies 
is discussed in \cite{VS16}. The decay of Maxwell's equations on Schwarzschild-de Sitter metric was proved recently in \cite{JK_17_1}. 

The case $\Lambda =0$ in \eqref{KdS1} and \eqref{ee} corresponds to the Kerr family of black holes, the stability of which is being pursued in a long-standing program that began soon after their discovery. 
In a remarkable recent development, the linear stability of Schwarzschild black hole spacetimes has been resolved in \cite{HDR_16} using the Teukolsky variables, by carrying forward the classic works in \cite{Regge-Wheeler_57, Zerilli_70, Vishveshvara_70,Moncrief_74}. The stability of Schwarzschild using metric coefficients has been resolved in \cite{HKW_16_1} and \cite{HKW_16_2}. A Morawetz estimate for the linearized gravity on Schwarzschild was proved in \cite{ABW_17}.  The decay of Maxwell's equations on the Schwarzschild metric was proved in \cite{PB_08}. Model problems for \emph{nonlinear} stability of  Schwarzschild are considered in \cite{Hto_16, KS_17, FP_17}. 

In contrast with Schwarzschild and Schwarzschild-de Sitter $(a =0)$, an important obstacle for Kerr and Kerr-de Sitter $(a \neq 0)$  is that the energy of even the linear wave equation is not necessarily positive-definite. This is caused by the ergo-region, which always surrounds a Kerr black hole with non-vanishing angular-momentum. 
A variety of techniques are introduced to cope with this issue for fields propagating on Kerr with small angular-momentum \cite{LB_15_1, LB_15_2, DR_11, Tato_11, Jluk_10}.

This leads us naturally to the question of stability of the Kerr-de Sitter family for large, but sub-extremal, angular-momentum $(\vert a\vert < m)$. It may be noted that  the lack of positivity of energy (and the related superradiance effect) makes proving the decay even more subtle for large $\vert a \vert$. The extraction of energy from a sub-extremal Kerr black hole using a linear wave equation is discussed in \cite{AAS_73, FKSY_08}, which is equivalent to the Penrose process \cite{DC_70}. The decay of the linear wave equation on a sub-extremal Kerr for fixed azimuthal modes is proved in \cite{FKSY_06, FKSY_08_E} using spectral methods \cite{FKSY_05}. Extending these works,  the decay for general solutions of the linear wave equation is proved in the intricate and remarkable work \cite{DRS_16}. However, little is known about the global behaviour of nonscalar and coupled fields propagating on Kerr or Kerr-de Sitter for large $\vert a \vert$.

The special case of axially symmetric linear waves, propagating on the Kerr-de Sitter spacetimes, admits a fortuitous simplification as the energy from the energy-momentum tensor is immediately positive-definite and is thus directly amenable to Morawetz and decay estimates (see, e.g., \cite{DR_10, LB_15_1} for the Kerr counterpart). The problem becomes much more subtle even for the axially symmetric coupled vector fields (e.g. Maxwell's equations) propagating on Kerr-de Sitter. Indeed the problem of positivity of total energy of axially symmetric Maxwell's equations on Kerr spacetimes has been an open problem for decades where, in principle, counter-examples for positivity of energy density can be constructed. This has recently been resolved for the full range of sub-extremal Kerr black holes in \cite{GM17} and separately in \cite{PW_17}. 

The subject of this paper is to prove equivalent results for Kerr-de Sitter $(\vert a \vert <m)$. Analogous to \cite{GM17}, these results also hold for the fully coupled axially symmetric Einstein-Maxwell perturbations of the Kerr-Newman-de Sitter spacetimes, which shall be discussed rigorously in a separate article. Importantly, in the pure Maxwell problem, the positive-definite energy we construct is naturally associated to \emph{gauge-invariant} quantities. In addition, without the several technicalities of the fully coupled Einstein-Maxwell problem, the pure Maxwell problem is more transparent.  

Following \cite{GM17}, we shall use the Hamiltonian formulation as it provides a mechanism to construct a gauge-invariant notion of mass-energy for the perturbative theory of $(\bar{M}, \bar{g})$ for the full $\vert a \vert <m$. 
Consider a Maxwell 2-form $F$ and a vector potential $A$ defined on $(\bar{M}, \bar{g})$, such that $F \fdg = dA$; then Maxwell's equations are the critical points of  the variational principle: 
\begin{align}
S_{M}[F] \fdg = -\frac{1}{4} \int \big\Vert F \big\Vert^2_{\bar{g}} \,\, \bar{\mu}_{\bar{g}}
\end{align}
\noindent for compactly supported variations.
If we perform the ADM decomposition 
\[ (\bar{M}, \bar{g}) = (\bar{\Sigma}, \bar{q}) \times \mathbb{R}\]
of the metric $\bar{g}$ and the vector potential $A$,
\begin{align}
\bar{g} =& - \bar{N}^2 dt^2 + \bar{q}_{ij} (dx^i + \bar{N}^i dt) \otimes (dx^j + \bar{N}^j dt) \\
A =& A_0 dt + A_i dx^i, \quad i, j=1, 2, 3,
\end{align}
and define $$\mathfrak{B}^i \fdg = \halb \eps^{ijk} (\ptl_j A_k- \ptl_k A_j),$$  the ADM variational principle is defined as 
\begin{align}\label{ADM-var}
I_{ADM}[A_i, \mathfrak{E}^i] \fdg =& \int \Big( A_i\ptl_t \mathfrak{E}^i - \halb \bar{N} \bar{
\mu}_{\bar{q}}^{-1}\bar{q}_{ij} ( \mathfrak{E}^i \mathfrak{E}^j +\mathfrak{B}^i \mathfrak{B}^j) +  \eps_{ijk} \bar{N}^i \mathfrak{E}^j \mathfrak{B}^k \notag\\ 
&\quad - A_0 \ptl_i \mathfrak{E}^i \Big) d^4x 
\end{align}
for the phase space $X^{\text{Max}}$, 
\[ X^{\text{Max}} \fdg = \{ (A_i, \mathfrak{E}^i), i= 1, 2, 3 \},\]
which results in the Maxwell field equations 
\begin{subequations} \label{adjoint}
\begin{align}
\ptl_t A_i =&- \bar{N} \bar{\mu}^{-1}_{\bar{q}} \bar{q}_{ij} \mathfrak{E}^j - \eps_{ijk} \bar{N}^j \mathfrak{B}^k + \ptl_iA_0, \\
\ptl_t \mathfrak{E}^i=& - \ptl_\ell (\bar{N} \bar{\mu}^{-1}_{\bar{q}} \bar{q}_{kj} \mathfrak{B}^j \eps^{k \ell i} ) +  \ptl_{\ell} (\bar{N}^{\ell} \mathfrak{E}^i - \bar{N}^i \mathfrak{E}^\ell ),\label{E-Eq}
\end{align}
\end{subequations}
where $\bar{\mu}_{\bar{q}}$ is the square root of the metric determinant of $(\bar{\Sigma}, \bar{q}).$
The $\mathfrak{E}^i$ field can be concisely represented in terms of the $F$ tensor as 
\begin{align}
\mathfrak{E}^i = \halb \eps^{ijk}\, \leftexp{*}{F}_{jk}.
\end{align}
The Maxwell constraint equations are 
\begin{align}\label{Max-const}
\ptl_i \mathfrak{E}^i =0, \quad i = 1, 2, 3.
\end{align}
\section{Dynamics with a Positive-Definite Hamiltonian}
Let $(\bar{M}, \bar{g})$ be the Kerr-de Sitter spacetime represented in \eqref{KdS1}, where $\Delta(r)$'s positive roots are $r_{\pm}$ and $r_c$ with $(r_\pm < r_c).$ If we consider the ADM decomposition of \eqref{KdS1},
\[ \bar{M} = \mathbb{R} \times \bar{\Sigma}, \]
where $\bar{\Sigma}$ is Riemannian.
The group $SO(2)$  acts on $(\bar{\Sigma}, \bar{q})$ in such a way that $\ptl_\phi$ is the associated Killing vector field.
  We define 
\begin{align}
\Sigma \fdg = \bar{\Sigma} / SO(2).
\end{align}
The fixed point set of the $SO(2)$ action on $\bar{\Sigma}$ is a union of two disjoint sets, which we represent together as $\Gamma$ (`the axes') for brevity. It may be noted that the fixed point set $\Gamma$ corresponds to
$\Vert \ptl_\phi \Vert_{\bar{g}} =0$ and also a boundary of $\Sigma.$ Finally, we define a Lorentzian manifold with boundary $M$ such that 
\[M \fdg= \bar{M}/ SO(2) = \Sigma \times \mathbb{R}. \]

\begin{proposition}
Suppose $(\bar{M}, \bar{g})$ is a Kerr-de Sitter spacetime with $\Delta$ as in \eqref{KdS1}.  Then the following statements hold. 
\begin{enumerate}[1]
\item[1.] The metric $\bar{g}$ can be represented in Weyl-Papapetrou form: 
\begin{align}
\bar{g} = e^{-2\gamma} \gg + e^{2\gamma} \Phi^2,
\end{align}
where $\Phi = d\phi + \mathcal{A}_\nu dx^\nu, \nu=0, 1, 2,$ $\gg$ is the Lorentzian metric of $M$.
\item[2.] There exists an auxiliary (scalar) potential $\omega$,
\[ \omega \fdg (M, \gg) \to \mathbb{R}\] such that $(\gamma, \omega)$satisfies the `shifted' wave maps equation: 
\begin{subequations}\label{swm}
\begin{align}
\square_\gg \gamma + \halb e^{-4\gamma}\gg^{\alpha \beta} \ptl_\alpha \omega \ptl_\beta \omega+ \Lambda e^{-2\gamma} =&0 \label{gammawm}\\
\square_\gg \omega - 4 e^{-4\gamma} \gg^{\a\b} \ptl_\a  \omega \ptl_\b \gamma =&0, \quad \text{on}  \quad(M, \gg) \setminus \Gamma,
\end{align}
\end{subequations}
we shall refer to $\omega$ as the gravitational twist potential. 
\item[3.] There exists a 3+1 decomposition of $(\bar{M}, \bar{g})$ such that it is smoothly foliated by 3D Riemannian maximal hypersurfaces. 
\end{enumerate} 
\end{proposition}
\begin{proof}
For proofs of statements 1 and 2 above, see \cite{NG_17_1}; the system \eqref{swm} is coupled to 2+1 Einstein equations. The fact that the expansion parameter $\Lambda$ decouples from the Einstein equations and appears as the forcing term of \eqref{gammawm} will play a crucial role in our problem. For the convenience of the reader, we shall provide the construction of the 
scalar potential $\omega$ below. The Einstein field equations \eqref{ee} imply that the 1-form $G$ such that 
\begin{align} \label{w-def} 
\mathcal{F}_{\mu \nu} = e^{-4\gamma} \varepsilon_{\mu \nu \delta} \gg^{\delta \a}G_\a,
\end{align}
 where $\varepsilon$ is the volume form of the metric $\gg$ and $\mathcal{F}_{\mu \nu}\fdg= \ptl_\mu \mathcal{A}_\nu - \ptl_\nu \mathcal{A}_\mu, \mu, \nu,\a,\delta = 0, 1, 2,$ is closed. Therefore, by the Poincar\'e Lemma, $G = d \omega,$ where $\omega$ is the gravitational twist potential.

 For (3), consider the  ADM decompositon of \eqref{KdS1}, 
\begin{align}
\bar{g} = - \bar{N}^2 dt^2 + \bar{q}_{ij} (dx^i + \bar{N}^i dt) \otimes (dx^j + \bar{N}^jdt), \quad i, j = 1, 2, 3.
\end{align}
Subsequently, if $\grad(\bar{q})$ is the (intrinsic) covariant derivative of $(\bar{\Sigma}, \bar{q})$, then it follows that
\begin{align}
\grad_i(\bar{q}) \bar{N}^i \equiv 0, \quad (\bar{\Sigma}_t, \bar{q}_t), \quad \forall t \in \mathbb{R},
\end{align}
which holds for all $t$ in view of the $t$-translational symmetry of $(\bar{M}, \bar{g})$ in \eqref{KdS1}. Furthermore, consider the ADM decomposition of $(M, \gg )$,
\begin{align}
\gg = - N^2 dt^2 + q_{ab} (dx^a + N^a dt) \otimes (dx^b + N^b dt), \quad a, b = 1, 2.
\end{align}
We also have 
\begin{align}
\grad_a (q) N^a \equiv 0, \quad (\Sigma_t, q_t), \quad \forall t \in \mathbb{R}.
\end{align}
\end{proof}
\noindent It follows from the definition \eqref{w-def} of $\omega$ that
\begin{align}\label{w-rel}
\ptl_a \mathcal{A}_0 + N e^{-4\gamma} \eps_{ab} \bar{\mu}_q q^{bc} \ptl_c \omega = 0.
\end{align}
\noindent The explicit expression of Kerr-de Sitter spacetime in the `Weyl-Papapetrou' form is as 
follows:
\begin{align}
\bar{g}_{KdS}  =& - e^{-2\gamma}\frac{\Pi \Delta \sin^2 \theta}{ (1+ \frac{\Lambda}{3}a^2)^4} dt^2
+ e^{-2\gamma}\, \left(\frac{\Sigma e^{2\gamma}}{\Delta} dr^2 + \frac{\Sigma e^{2\gamma}}{\Pi} d\theta^2 \right) \notag\\
& + e^{2\gamma} \left(d\phi + \frac{a ( -2mr - \frac{\Lambda}{3}a^2 (r^2 + a^2)(1+ \cos^2 \theta))}{ - a^2 \sin^2 \theta \Delta +   \Pi(r^2 + a^2)^2} dt \right)^2.
\end{align}
Reading off various components of the Weyl-Papapetrou form, we have
\begin{subequations}
\begin{align}
e^{2\gamma} =& \frac{\sin^2 \theta (-a^2 \sin^2 \theta \Delta + \Pi (r^2 + a^2)^2)}{\Sigma (1+ \frac{\Lambda}{3} a^2)^2}, \\
N=& \frac{(\Pi \Delta)^{\halb} \sin \theta}{(1+ \frac{\Lambda}{3}a^2)^2}, \\
\mathcal{A}_0=& \frac{a( -2mr - \frac{\Lambda}{3} a^2 (r^2 +a^2) (1+ \cos^2 \theta))}{-a^2 \sin^2 \theta \Delta
+ \Pi(r^2 + a^2)^2},\\
\bar{\mu}_q^{-1} q_{ab} dx^a \otimes dx^b =& \left(\frac{\Pi}{\Delta}\right)^{\halb} dr^2 + \left( \frac{\Delta}{\Pi} \right)^{\halb} d\theta^2,
\end{align}
\end{subequations}
where $\bar{\mu}_q$ is the square root of the metric determinant of $(\Sigma, q).$ We are interested in the initial value  problem of  Maxwell's equations \eqref{adjoint} with axial symmetry. We assume that  the axially symmetric $F$ tensor is derivable from an axially symmetric vector potential $A$.

In view of the fact that the Kerr-de Sitter spacetime is also axially symmetric, let us construct a new phase space $X$. Firstly, consider a twist potential $$\lambda \fdg (M, \gg) \to \mathbb{R}$$
such that $\lambda \fdg = A_{\phi}$, so that  $\mathfrak{B}^a = \eps^{ab} \ptl_b \lambda$ and $\ptl_a \mathfrak{B}^a =0, a, b = 1, 2.$ It follows from the Maxwell constraint equations $\ptl_a \mathfrak{E}^a =0$ and the Poincar\'e Lemma that there exists a twist potential $$\eta \fdg (M, \gg) \to \mathbb{R}$$ such that $\mathfrak{E}^a = \eps^{ab}\ptl_b \eta.$ Likewise, it follows from the variational principle \eqref{ADM-var} that
 the conjugate momenta $u$ and $v$ defined as 
\begin{align}
u \fdg = \mathfrak{B}^\phi, \quad v \fdg = - \mathfrak{E}^{\phi}
\end{align}
form the dynamical canonical pairs with $\eta$ and $\lambda$ respectively. Thus we define the phase space $X$ as
\[X \fdg=\{ (\lambda, v), (\eta, u) \}.\]
 Maxwell's equations \eqref{adjoint}, can be transformed into the phase space $X$  and locally represented as 
\begin{subequations}\label{feq}
\begin{align}
&\ptl_t \eta = N e^{2\gamma} \bar{\mu}^{-1}_q u, \quad \ptl_t \lambda= N e^{2\gamma} \bar{\mu}^{-1}_q v, \\
&\ptl_t u = \ptl_b (N \bar{\mu}_q q^{ab} e^{-2\gamma} \ptl_a \eta ) + N \bar{\mu}_q q^{ab} e^{-4\gamma} \ptl_a \omega \ptl_b \lambda, \\
&\ptl_t v= \ptl_b (N \bar{\mu}_q q^{ab} e^{-2\gamma} \ptl_a \lambda) - N \bar{\mu}_q q^{ab}e^{-4\gamma} \ptl_a \omega \ptl_b \eta.
\end{align}
\end{subequations}
For the initial value problem of \eqref{feq}, we assume that the axially symmetric $F$ tensor has smooth and compactly supported initial data in a $t=t_0$ initial data slice $(\bar{\Sigma}_0, \bar{q}_0)$. 
Define initial data in $X$ as
\begin{align}
ID \fdg =\{ (\lambda_0, u_0), (\eta_0, v_0)\}, \quad (\Sigma_0, q_0).
\end{align} 
In this work, we shall assume that the initial data is compactly supported, strictly within the  interior of $\Sigma$, i.e., $Supp(ID) \subset \Sigma$. As a consequence, in the computations the boundary terms vanish at both the horizons. 
\noindent It may be noted that the global propagation of regularity of the Maxwell field $F$ in the \emph{domain of outer communications} of the Kerr-de Sitter metric $(\bar{M}, \bar{g})$ is standard. As a consequence, we have the following prescribed behaviour on the axes $\Gamma$ of the quotient space $(\Sigma, q):$
\begin{subequations}\label{axisbehav}
\begin{align}
\ptl_{\vec{s}} \lambda =0,&\,\, \ptl_{\vec{s}} \eta =0,\quad\, \text{on}\quad \Gamma, \\
\ptl_{\vec{n}} \lambda=0,&\,\, \ptl_{\vec{n}} \eta =0, \quad \text{on} \quad \Gamma, \\
\ptl_t \lambda=0,&\,\,\ptl_t \eta =0,\, \quad \text{on} \quad \Gamma, \quad \forall\, t \in \mathbb{R},
\end{align}
\end{subequations}
where $\ptl_{\vec{s}}$ and $\ptl_{\vec{n}}$ are the derivatives tangential and normal to the axes $\Gamma$ respectively. In view of \eqref{axisbehav} it may be noted that 
\[\mathfrak{E}^i =0, \quad \mathfrak{B}^i =0, \quad \text{on} \quad \Gamma.\]
It follows from \eqref{axisbehav} that one can choose $\lambda, \eta$ such that they are (uniformly) $0$ along $\Gamma.$ In principle, our Hamiltonian framework allows for the Coulomb type conserved charges; however, the behaviour at the axes is chosen only for convenience in functional analysis arguments. We now state the main theorem of the paper.

\begin{theorem}
Suppose $F$ is the electromagnetic  Faraday tensor with $\mathcal{L}_{\ptl_\phi} F \equiv 0,\, \mathcal{L}_{\ptl_\phi} A \equiv 0,$ propagating on the Kerr-de Sitter black holes \eqref{KdS1} with $\vert a \vert < m$ and further suppose that $F \in C^{\infty}(\bar{\Sigma}_0, \bar{q}_0)$ with $Supp(ID) \subset \Sigma_0$. Then the following statements hold for the initial value problem of $F$ on Kerr-de Sitter $(\bar{M}, \bar{g})$
\begin{enumerate}[1]

\item[1.] There exists a positive-definite Hamiltonian $H^{\text{Alt}}$ for the dynamics of the canonical pairs in the phase space $X = \{ (\lambda, u),(\eta, v) \}$ i.e., 
\begin{subequations}
\begin{align}
D_{u} \cdot H^{\text{Alt}} = \ptl_t \eta, \quad D_{\eta} \cdot H^{\text{Alt}} = - \ptl_t u , \\
D_{v} \cdot H^{\text{Alt}} = \ptl_t \lambda, \quad D_{\lambda} \cdot H^{\text{Alt}} = - \ptl_t v,
\end{align}
\end{subequations}
where $D$ is the (variational) directional derivative in the phase-space $X$.
\item[2.] There exits a divergence-free spacetime vector density $J$ such that its flux through $t$-constant hypersurfaces is positive-definite. 
\item[3.]  There exists a canonical transformation $U \fdg (X, H^{\text{Alt}}) \to (\ulin{X}, H^{\text{Reg}})$ to a `regularized' phase space $$\ulin{X} \fdg =
 \{ (\ulin{\lambda}, \ulin{u}),(\ulin{\eta}, \ulin{v}) \}$$ where, 
 $$\ulin{\lambda} \fdg= e^{-\gamma} \lambda,\quad \ulin{\eta} \fdg = e^{-\gamma} \eta,\quad \ulin{u} \fdg = e^{\gamma} u,\quad \ulin{v} \fdg = e^{\gamma} v,$$ such that the corresponding Hamiltonian $H^{\text{Reg}}$ is positive-definite.
\end{enumerate}
\end{theorem}
\begin{proof}
The ADM Hamiltonian energy, re-expressed in the phase space $X$ and using \eqref{w-rel}, consecutively transforms as follows:
\begin{align} \label{origham}
H =& \int_{\Sigma} \Big(\halb N \bar{\mu}^{-1}_{q} e^{2\gamma}( u^2  + v^2) + \halb N \bar{\mu}_q q^{ab} e^{-2\gamma} (\ptl_a \eta \ptl_b \eta + \ptl_a \lambda \ptl_b ) \notag\\ 
&\quad-  \mathcal{A}_0 \eps^{ab} \ptl_a \eta \ptl_b \lambda \Big) d^2 x \notag\\
=& \int_{\Sigma}  \Big(\halb N \bar{\mu}^{-1}_q e^{2\gamma}(u^2 + v^2) + \halb N \bar{\mu}_q q^{ab} e^{-2\gamma} (\ptl_a \eta \ptl_b \eta + \ptl_a \lambda \ptl_b \lambda) \notag\\
&\quad + N e^{-4\gamma} \bar{\mu}_q q^{ab} \ptl_a \omega \ptl_b \eta \lambda \Big) d^2 x
\end{align}
where the $-v \eps^{ab} \ptl_b\eta + u \eps^{ab} \ptl_b \lambda$ terms drop out. Now consider the quantity
\begin{align*}
I \fdg =&\frac{1}{4} N e^{-2\gamma} \bar{\mu}_q q^{ab}  \left(( \ptl_a \lambda -2\lambda \ptl_a\gamma) (\ptl_b\lambda - 2\lambda \ptl_b \gamma) + ( \ptl_a \eta -2\eta \ptl_a\gamma) (\ptl_b\eta - 2\eta\ptl_b \gamma) \right) \notag\\
&+ \frac{1}{4}N e^{-2\gamma} \bar{\mu}_q q^{ab} \big(  (\ptl_a\eta + \lambda e^{- 2\gamma} \ptl_a \omega)( \ptl_b \eta + \lambda e^{-2\gamma} \ptl_b \omega) \\
&+ ( \ptl_a \lambda - \eta e^{-2\gamma} \ptl_a \omega)( \ptl_b \lambda - \eta e^{-2\gamma} \ptl_b \omega) \big) 
- \halb Ne^{-2\gamma}  \bar{\mu}_q q^{ab} (\ptl_a \lambda \ptl_b \lambda + \ptl_a \eta \ptl_b \eta).
\end{align*}
We have, 
\begin{align}
I =& N e^{-2\gamma} \bar{\mu}_q q^{ab} (\ptl_a \gamma \ptl_b \gamma+ \frac{1}{4} e^{-4\gamma} \ptl_a \omega\ptl_b \omega) (\lambda^2 + \eta^2) \notag\\
&-Ne^{-2\gamma}\bar{\mu}_q q^{ab} (\lambda \ptl_a \gamma \ptl_b \lambda + \eta \ptl_a \gamma \ptl_b \eta - \halb \lambda e^{-2\gamma}\ptl_a\omega \ptl_b \eta + \halb \eta e^{-2\gamma} \ptl_a \omega \ptl_b \lambda ).
\end{align}
Recall the wave map system satisfied by $(\gamma, \omega)$: 
\begin{subequations}\label{wavemapKdS}
\begin{align}
\ptl_b (N \bar{\mu}_q q^{ab} \ptl_a \gamma) + \halb N \bar{\mu}_q e^{-4\gamma} q^{ab} \ptl_a  \omega \ptl_b\omega + N \bar{\mu}_q \Lambda e^{-2\gamma} =&0 ,\\
\ptl_b (N \bar{\mu}_q q^{ab} e^{-4\gamma} \ptl_a \omega)=&0.
\end{align}
\end{subequations}
\noindent Now consider the quantity: 
\begin{align}
II \fdg = & \halb \ptl_b \left( - N \bar{\mu}_q q^{ab} e^{-4\gamma} \ptl_a \omega \eta \lambda - N \bar{\mu}_q q^{ab} e^{-2\gamma}\ptl_a \gamma (\eta^2 + \lambda^2) \right) .
\end{align}
We have 
\begin{align*}
II =&-\halb  \big( N \bar{\mu}_q q^{ab} e^{-4\gamma} \ptl_a \omega \ptl_b \eta \lambda + \lambda \eta \ptl_b (N \bar{\mu}_q q^{ab} e^{-4\gamma} \ptl_a \omega) \\
&+ N \bar{\mu}_q q^{ab} e^{-4\gamma} \ptl_a \omega \ptl_b \lambda \eta \big)  \notag\\
&- \halb \big( e^{-2\gamma} \ptl_b (N \bar{\mu}_q q^{ab} \ptl_a \gamma) (\lambda^2 + \eta^2)  -2e^{-2\gamma}N \bar{\mu}_q q^{ab} \ptl_a \gamma \ptl_b \gamma ( \lambda^2 + \eta^2) \\
&-  2e^{-2\gamma}  N \bar{\mu}_q q^{ab} \ptl_a \gamma (\lambda \ptl_b \lambda + \eta \ptl_b \eta)  \big) \\
=&- \halb \left( N \bar{\mu}_q q^{ab} e^{-4\gamma} \ptl_a \omega \ptl_b \eta \lambda  + N \bar{\mu}_q q^{ab} e^{-4\gamma} \ptl_a \omega \ptl_b \lambda \eta \right)  \\
& - \halb ( e^{-2\gamma} (-\halb N \bar{\mu}_q e^{-4\gamma} q^{ab} \ptl_a  \omega \ptl_b\omega - N \bar{\mu}_q \Lambda e^{-2\gamma}) \\
&-2e^{-2\gamma} N \bar{\mu}_q q^{ab} \ptl_a \gamma \ptl_b \gamma ) ( \lambda^2 + \eta^2) 
-  e^{-2\gamma}  N \bar{\mu}_q q^{ab} \ptl_a \gamma (\lambda \ptl_b \lambda + \eta \ptl_b \eta).  
\end{align*}
\noindent Consequently, 
\begin{align}
I - II =   -\halb (\lambda^2 + \eta^2)\Lambda e^{-4\gamma} N \bar{\mu}_q + N \bar{\mu}_q q^{ab} e^{-4\gamma} \ptl_a \omega \ptl_b \eta \lambda,
\end{align}
as the $\ptl \omega \ptl \eta$ term occurs in both $I$ and $II$ with opposite signs.  
Therefore, using I--II, we can transform the Hamiltonian energy \eqref{origham} into the following manifestly positive form: 
\begin{align}\label{Hampos}
H^{\text{Alt}} \fdg=& \int_\Sigma \Big( \halb N e^{2\gamma} \bar{\mu}^{-1}_q (u^2 + v^2) + \halb \Lambda  N \bar{\mu}_q e^{-4\gamma} (\lambda^2 + \eta^2) \notag\\
&+ \frac{1}{4}N e^{-2\gamma} \bar{\mu}_q q^{ab} \big((\ptl_a \lambda -2\lambda \ptl_a\gamma) (\ptl_b \lambda-2\lambda \ptl_b \gamma) \notag\\
&+ (\ptl_a \eta -2\eta \ptl_a\gamma) (\ptl_b \eta-2\eta \ptl_b \gamma) \big) \notag\\
&+ \frac{1}{4} N e^{-2\gamma} \bar{\mu}_q q^{ab} \big(( \ptl_a \eta + \lambda e^{-2\gamma} \ptl_a \omega)(\ptl_b \eta + \lambda e^{-2\gamma} \ptl_b \omega) \notag \\&+ (\ptl_a \lambda - \eta e^{-2\gamma} \ptl_a \omega)(\ptl_b \lambda - \eta e^{-2\gamma} \ptl_b \omega) \big) \Big) d^2 x.
\end{align}
\noindent The aforementioned transformation of the original ADM Hamiltonian into a positive form in \eqref{Hampos} is motivated by the construction of the Robinson's identity \cite{Rob_74}, but now adapted to our problem. Crucially, we need to prove that the Hamiltonian structure of the equations is retained.

Consider a (variational) $1$-parameter flow of a generic phase point $P$ in the phase space $X$, parametrized by $s$. We shall denote the components of the variation at $P$ with respect to this flow as 
\begin{align}
u' \fdg = D \cdot u (P), \quad v' \fdg = D \cdot v(P), \quad \lambda' \fdg = D \cdot \lambda (P), \quad \eta' \fdg = D \cdot \eta (P).
\end{align}
\noindent Consider $D_u \cdot H^{\text{Alt}}$ and $D_v \cdot H^{\text{Alt}}$; we have \[D_u \cdot H^{\text{Alt}} = N e^{2\gamma} \bar{\mu}^{-1}_q u \quad \text{and} \quad  D_v \cdot H^{\text{Alt}} = N e^{2\gamma} \bar{\mu}^{-1}_q v, \] respectively. The quantities $D_\lambda \cdot H^{\text{Alt}}$ and $D_\eta \cdot H^{\text{Alt}}$ are more difficult. 
From \eqref{Hampos}, $D_\eta \cdot H^{\text{Alt}}$ and $D_\lambda \cdot H^{\text{Alt}}$ have the following types of terms: 

\begin{itemize}
\item 1st order $\ptl \eta \ptl \eta'$ and $\ptl \lambda \ptl \lambda'$: 
\begin{align*}
N e^{-2\gamma} \bar{\mu}_q q^{ab} \ptl_a \eta \ptl_b \eta' =& \ptl_b ( N e^{-2\gamma} \bar{\mu}_q q^{ab} \ptl_a \eta \eta') - 
\ptl_b ( N e^{-2\gamma}\bar{\mu}_q q^{ab} \ptl_b \eta) \eta'
\intertext{and}
N e^{-2\gamma} \bar{\mu}_q q^{ab} \ptl_a \lambda \ptl_b \lambda'
=& \ptl_b (N e^{-2\gamma} \bar{\mu}_q q^{ab} \ptl_a \lambda \lambda') - 
\ptl_b (N e^{-2\gamma} \bar{\mu}_q q^{ab} \ptl_a \lambda) \lambda';
\end{align*}
\item 1st order $\ptl \eta' \ptl \omega$ and $\ptl\lambda' \ptl \omega$:
\begin{align*}
\halb \lambda N e^{-4\gamma} \bar{\mu}_q q^{ab} \ptl_a \eta' \ptl_b \omega =&
\ptl_b (\halb \lambda N e^{-4\gamma} \bar{\mu}_q q^{ab} \ptl_a \omega \eta') - 
\ptl_b ( \halb \lambda N e^{-4\gamma} \bar{\mu}_q q^{ab} \ptl_b \omega) \eta'
\intertext{and}
-\halb N e^{-4\gamma} \bar{\mu}_q q^{ab} \eta \ptl_a \lambda' \ptl_b \omega=&
-\halb \ptl_b ( \eta N e^{-4\gamma} \bar{\mu}_q q^{ab} \ptl_a \omega \lambda') + \halb
\ptl_b ( \eta N e^{-4\gamma} \bar{\mu}_q q^{ab} \ptl_a \omega) \lambda';
\end{align*}
\item mixed type $\eta' \ptl \eta$, $\eta \ptl \eta' $ and $\lambda' \ptl \lambda$, $\lambda \ptl \lambda'$: 
\begin{align*}
&-N e^{-2\gamma} \bar{\mu}_q q^{ab} \ptl_a \eta \ptl_b \gamma \eta' - N e^{-2\gamma} \bar{\mu}_q q^{ab} \eta \ptl_a \eta' \ptl_b \gamma \notag\\
&= -N e^{-2\gamma} \bar{\mu}_q q^{ab} \ptl_a \eta \ptl_b \gamma \eta' - \ptl_b (N e^{-2\gamma} \bar{\mu}_q q^{ab} \eta \ptl_a \gamma \eta') + \ptl_b (N e^{-2\gamma} \bar{\mu}_q q^{ab} \eta \ptl_b \gamma) \eta'
\intertext{and}
&-Ne^{-2\gamma} \bar{\mu}_q q^{ab} \lambda' \ptl_a \lambda \ptl_b \gamma  - N e^{-2\gamma} \bar{\mu}_q q^{ab} \lambda \ptl_a \lambda' \ptl_b \gamma \notag\\
=&  -Ne^{-2\gamma} \bar{\mu}_q q^{ab} \lambda' \ptl_a \lambda \ptl_b \gamma
-\ptl_b (e^{-2\gamma} N \bar{\mu}_q q^{ab}\lambda \ptl_b \gamma \lambda') + \ptl_b (e^{-2\gamma} N \bar{\mu}_q q^{ab}\lambda \ptl_b \gamma) \lambda';
\end{align*}
\item 0th order: 
\begin{align}
2(N \bar{\mu}_q e^{-2\gamma} q^{ab}  \ptl_a \gamma \ptl_b \gamma + \frac{1}{4} N \bar{\mu}_q e^{-6\gamma} q^{ab} \ptl_a \omega \ptl_b \omega + \halb \Lambda N \bar{\mu}_q e^{-4\gamma})\eta \eta'
\intertext{and}
2 (N \bar{\mu}_q e^{-2\gamma} q^{ab} \ptl_a \gamma \ptl_b \gamma + \frac{1}{4} N \bar{\mu}_q e^{-6\gamma} q^{ab} \ptl_a \omega \ptl_b \omega + \halb \Lambda N \bar{\mu}_q e^{-4\gamma}) \lambda \lambda'.
\end{align}
\end{itemize}
Combining all the above, while using the system \eqref{wavemapKdS} again, we recover the full set of  field equations:
\begin{align}
D_ \eta \cdot H^{\text{
Alt}} =& - \ptl_b (N e^{-2\gamma} \bar{\mu}_q q^{ab}\ptl_b \eta) - N \bar{\mu}_q q^{ab} e^{-4\gamma}\ptl_a\omega \ptl_b \lambda \notag\\
=&\, - \ptl_t u 
\intertext{and}
D_\lambda \cdot H^{\text{Alt}} =& -\ptl_b (N e^{-2\gamma} \bar{\mu}_q q^{ab}\ptl_b \lambda) + N \bar{\mu}_q q^{ab} e^{-4\gamma}\ptl_a\omega \ptl_b \eta \notag\\
=&\,  -\ptl_t v.
\end{align} 
Now let us turn to Part 2. Define the energy density $\mathcal{E}^{\text{Alt}}$ of the Hamiltonian $H^{\text{Alt}}$ such that
\begin{align}
H^{\text{Alt}} = \int_{\Sigma} \mathcal{E}^{\text{Alt}}\, d^2 x.
\end{align}
Now define $ \bar{v} = N\bar{\mu}^{-1}_q v$ and $\bar{u} \fdg = N \bar{\mu}^{-1}_q u$.
We shall construct
  the divergence-free vector density from the time derivative of the density $\mathcal{E}^{\text{Alt}}$ of the  Hamiltonian $H^{\text{Alt}} $. The purpose of calculating the time derivative
  of the energy density is to obtain a pure spatial divergence.  Therefore, we collect the terms with $\bar{v}$ and $\bar{u}$ and their spatial derivatives $(\ptl_a\bar{v}, \ptl_a\bar{u})$ separately, both of which occur.  It turns out that these terms combine to form a pure patial divergence, if we use the background field equations.  Explicitly, we have the following terms in
 $\frac{\ptl}{\ptl t} \mathcal{E}^{\text{Alt}}$:
\begin{align}
\ptl_a \bar{v} (\halb N \bar{\mu}_q q^{ab} (2\ptl_b \lambda - \eta e^{-2\gamma}\ptl_b \omega -2  \lambda \ptl_b \lambda))
\intertext{and}
\ptl_a \bar{u} (\halb N \bar{\mu}_q q^{ab} (2\ptl_b \eta + \lambda e^{-2\gamma}\ptl_b \omega -2 \eta \ptl_b \gamma )),
\end{align} 
\noindent where the terms with $\bar{v}$ and $\bar{u}$ are 
\begin{align}
&\bar{u} (\halb N e^{-2\gamma} \bar{\mu}_q q^{ab} (-\ptl_a \omega (\ptl_b \lambda- \eta e^{-2\gamma} \ptl_b\omega) + 2 e^{2\gamma} \ptl_a \gamma (\ptl_b \eta + \lambda e^{-2\gamma} \ptl_b \omega) )) \notag\\
& + \bar{u} e^{2\gamma}( \ptl_b (N\bar{\mu}_q q^{ab}e ^{-2\gamma}\ptl_b \eta)+ N \bar{\mu}_q q^{ab} e^{-4\gamma}\ptl_a \omega \ptl_b \lambda + N \lambda
e^{-4\gamma}\bar{\mu}_q \Lambda) 
\intertext{and}
 &\bar{v} (\halb N e^{-2\gamma} \bar{\mu}_q q^{ab} (\ptl_a \omega (\ptl_b \eta +\lambda e^{-2\gamma} \ptl_b \omega ) + 2 e^{2\gamma} \ptl_a \gamma(\ptl_b \lambda - \eta e^{-2\gamma} \ptl_b \omega))) \notag\\
 &+ \bar{v} e^{2\gamma} (\ptl_b(N \bar{\mu}_q q^{ab}e^{-2\gamma}\ptl_a \lambda) - N \bar{\mu}_q q^{ab}e^{-4\gamma} \ptl_a \omega \ptl_b \eta + N \eta e^{-4\gamma} \bar{
 \mu}_q \Lambda)
\end{align}
which, in view of the system \eqref{wavemapKdS}, can be transformed to 
\begin{align}
\bar{v} \ptl_b (N \bar{\mu}_q q^{ab} ( \ptl_a \eta - \halb \eta e^{-2\gamma} \ptl_a \omega - \lambda \ptl_b \gamma))\\
\intertext{and}
\bar{u} \ptl_b (N \bar{\mu}_q q^{ab} (\ptl_a \eta + \halb \lambda e^{-2\gamma} \ptl_a \omega - \eta \ptl_a \gamma)),
\end{align}
respectively. Therefore, the time derivative of $\mathcal{E}^{\text{Alt}}$ can be transformed into a pure spatial divergence: 
\begin{align}
\frac{\ptl}{\ptl t} \mathcal{E}^{\text{Alt}} =& \frac{\ptl}{\ptl x^b} \big( \bar{u} N \bar{\mu}_q q^{ab} (\ptl_a \eta + \halb \lambda e^{-2\gamma} \ptl_a \omega - \eta \ptl_a \gamma ) \notag\\
&+ \bar{v} N \bar{\mu}_q q^{ab} (\ptl_a \lambda - \halb \eta e^{-2\gamma} \ptl_a \omega - \lambda \ptl_a \gamma) \big) \notag\\
=& \frac{\ptl}{\ptl x^b} \big( u N^2  q^{ab} (\ptl_a \eta + \halb \lambda e^{-2\gamma} \ptl_a \omega - \eta \ptl_a \gamma ) \notag\\
&+ v N^2  q^{ab} (\ptl_a \lambda - \halb \eta e^{-2\gamma} \ptl_a \omega - \lambda \ptl_a \gamma)\big)
\end{align}
which can be transformed into a divergence-free vector density
\begin{align}
J \fdg= J^t \ptl_t + J^b \ptl_b
\end{align}
where 
\begin{align*}
J^t \fdg=&\, \mathcal{E}^{\text{Alt}} \notag\\
J^b  \fdg=&\, -u N^2  q^{ab} (\ptl_a \eta + \halb \lambda e^{-2\gamma} \ptl_a \omega - \eta \ptl_a \gamma ) - v N^2  q^{ab} (\ptl_a \lambda - \halb \eta  e^{-2\gamma} \ptl_a \omega - \lambda \ptl_a \gamma).
\end{align*}
For  Part 3, consider the regularized phase space $\ulin{X} \fdg = \{ ( \ulin{\lambda}, \ulin{v}), (\ulin{\eta}, \ulin{u})  \}$
\[ \ulin{\gamma} \fdg = e^{-\gamma} \gamma,\quad \ulin{\eta} \fdg = e^{-\gamma} \eta,\quad \ulin{u} \fdg= e^{\gamma}u, \quad \ulin{v} \fdg= e^{\gamma} v\]
To construct a regularized Hamiltonian $H^{\text{Reg}}$, we shall use further identities: 
\begin{align*}
 &\hspace*{-2em}\frac{1}{4}N e^{-2\gamma} \bar{\mu}_q q^{ab} \big((\ptl_a \lambda -2\lambda \ptl_a\gamma) (\ptl_b \lambda-2\lambda \ptl_b \gamma) + (\ptl_a \eta -2\eta \ptl_a\gamma) (\ptl_b \eta-2\eta \ptl_b \gamma) \big) \notag\\
  &+ \frac{1}{4} N e^{-2\gamma} \bar{\mu}_q q^{ab} \big(( \ptl_a \eta + \lambda e^{-2\gamma} \ptl_a \omega)(\ptl_b \eta + \lambda e^{-2\gamma} \ptl_b \omega) \notag \\
  &+ (\ptl_a \lambda - \eta e^{-2\gamma} \ptl_a \omega)(\ptl_b \lambda - \eta e^{-2\gamma} \ptl_b \omega) \big)  \notag\\
  =\quad& \frac{1}{4} N \bar{\mu}_q q^{ab} ((\ptl_a \ulin{\lambda} -\ulin{\lambda} \ptl_a \gamma)( \ptl_b \ulin{\lambda} - \ulin{\lambda} \ptl_b \gamma) + (\ptl_b \ulin{\eta} -\ulin{\eta} \ptl_a \gamma)(\ptl_b \ulin{\eta} - \ulin{\eta} \ptl_b \gamma)) \\
  &+ \frac{1}{4} N \bar{\mu}_q q^{ab} ((\ptl_a \ulin{\eta} + \ulin{\eta} \ptl_a \gamma + \ulin{\lambda} e^{-2\gamma} \ptl_a \omega)(\ptl_b \ulin{\eta} + \ulin{\eta} \ptl_b \gamma + \ulin{\lambda} e^{-2\gamma} \ptl_b \omega)) \notag\\
  &+ \frac{1}{4} N \bar{\mu}_q q^{ab} ((\ptl_a \ulin{\lambda} + \ulin{\lambda} \ptl_a \gamma  - \ulin{\eta} e^{-2\gamma} \ptl_a \omega)(\ptl_b \ulin{\lambda} + \ulin{\lambda} \ptl_b \gamma  -\ulin{\eta} e^{-2\gamma} \ptl_b \omega)) \notag\\
  =\quad& \halb N \bar{\mu}_q q^{ab} ( (\ptl_a \ulin{\lambda} - \halb \ulin{\eta} e^{-2\gamma} \ptl_a \omega) (\ptl_b \ulin{\lambda} - \halb \ulin{\eta} e^{-2\gamma} \ptl_b \omega) \\
  &+ (\ptl_a \ulin{\eta} + \halb \ulin{\lambda} e^{-2\gamma} \ptl_a \omega)(\ptl_b \ulin{\eta} + \halb \ulin{\lambda} e^{-2\gamma} \ptl_b \omega)) \notag\\
  &+ \halb N \bar{\mu}_q q^{ab}( \ptl_a \gamma \ptl_b \gamma +  \frac{1}{4} e^{-4\gamma} \ptl_a \omega \ptl_b \omega) (\ulin{\lambda}^2 + \ulin{\eta}^2),
\end{align*}
where the $\ptl \omega \ptl \gamma$ terms cancel.  Thus, the energy $H^{\text{Alt}}$ can be transformed into 
\begin{align}
H^{\text{Reg}} \fdg =& \int  \Big(\halb N \bar{\mu}_q^{-1} (\ulin{u}^2 + \ulin{v}^2) + \halb N \Lambda \bar{\mu}_q e^{-2\gamma} (\ulin{\lambda}^2 + \ulin{\eta}^2)\notag \\
&+ \halb N \bar{\mu}_q q^{ab} ( \ptl_a \gamma \ptl_b \gamma +  \frac{1}{4} e^{-4\gamma} \ptl_a \omega \ptl_b \omega) (\ulin{\lambda}^2 + \ulin{\eta}^2)  \notag\\
&+\halb N \bar{\mu}_q q^{ab} ( (\ptl_a \ulin{\lambda} - \halb \ulin{\eta} e^{-2\gamma} \ptl_a \omega) (\ptl_b \ulin{\lambda} - \halb \ulin{\eta} e^{-2\gamma} \ptl_b \omega) \notag\\
&+ (\ptl_a \ulin{\eta} + \halb \ulin{\lambda} e^{-2\gamma} \ptl_a \omega)(\ptl_b \ulin{\eta} + \halb \ulin{\lambda} e^{-2\gamma} \ptl_b \omega)) \Big) d^2 x.\notag\\
\end{align}
Analogously to the calculations for $H^{\text{Alt}}$, we recover the field equations for the regularized phase space $\ulin{X}$:  
\begin{subequations}\label{Reg-Ham-Eq}
\begin{align}
\ulin{D}_{\ulin{u}} \cdot H^{\text{Reg}} = \ptl_t \ulin{\eta}, \quad \ulin{D}_{\ulin{\eta}} \cdot H^{\text{Reg}} = - \ptl_t \ulin{u},  \\
\ulin{D}_{\ulin{v}} \cdot H^{\text{Reg}} = \ptl_t \ulin{\lambda}, \quad \ulin{D}_{\ulin{\lambda}} \cdot H^{\text{Reg}} = - \ptl_t \ulin{v},
\end{align}
\end{subequations}
\noindent where $\ulin{D}$ is the usual (variational) directional derivative in $\ulin{X}.$ Likewise, if we define the energy density $\mathcal{E}^{\text{Reg}}$ such that 
\begin{align}
H^{\text{Reg}} = \int_{\Sigma} \mathcal{E}^{\text{Reg}}\,\, d^2 x
\end{align}
and time differentiate, using the system \eqref{Reg-Ham-Eq} we get: 
\begin{align}
\frac{\ptl}{\ptl t} \mathcal{E}^{\text{Reg}} = \frac{\ptl}{\ptl x^b} ( N^2 q^{ab} \ulin{u} ( \ptl_a \ulin{\eta} + \halb \ulin{\lambda} e^{-2\gamma} \ptl_a \omega ) + N^2 q^{ab} \ulin{v} ( \ptl_a \ulin{\lambda} - \halb \ulin{\eta} e^{-2\gamma} \ptl_a \omega) ),
\end{align}
\noindent which results in a vector field density
\begin{align}
J^{\text{Reg}} \fdg= (J^{\text{Reg}})^t \ptl_t + (J^{\text{Reg}})^b \ptl_b,
\end{align}
where 
\begin{align}
(J^{\text{Reg}})^t =& \mathcal{E}^{\text{Reg}}, \notag\\
(J^{\text{Reg}})^b=& -N^2 q^{ab} \ulin{u} (\ptl_a \ulin{\eta} + \halb \ulin{\lambda} e^{-2 \gamma} \ptl_a \omega) -
N^2 q^{ab} \ulin{v} ( \ptl_a \ulin{\lambda} - \halb \ulin{\eta} e^{-2\gamma} \ptl_a \omega),
\end{align}
which is also divergence free. The advantage of recasting in the $(\ulin{X}, H^{\text{Reg}})$ framework is that it has better behaviour on the axes and the horizon than $(X, H^{\text{Alt}})$.  
\end{proof}

\subsection*{Acknowledgements} 
The author is grateful to Vincent Moncrief for  enjoyable interactions and  feedback. A part of this work was done when the author was supported by DFG Fellowship GU 1513/1-1. The author is also grateful to the (anonymous) referee for constructive comments that improved the article. 
\bibliographystyle{plain}

\end{document}